\documentclass[12pt,a4paper]{amsart}
\usepackage{latexsym,amssymb,amsmath,amsthm,amsfonts,enumerate,verbatim,xspace,
exscale}
\usepackage{graphicx}
\usepackage{color,amsbsy}
 
\usepackage{hyperref}
\definecolor{darkblue}{rgb}{0,0,.5}
\hypersetup{colorlinks=true, breaklinks=true, linkcolor=darkblue, menucolor=darkblue,  urlcolor=darkblue,citecolor=darkblue}

\usepackage{setspace}

\theoremstyle{plain}
\newtheorem{theorem}{Theorem}[section]
\newtheorem{lemma}[theorem]{Lemma}
\newtheorem{proposition}[theorem]{Proposition}

\theoremstyle{definition}
\newtheorem{definition}[theorem]{Definition}

\newcommand{\Section}[1]{\section{#1} \setcounter{equation}{0}}

\def\d{\textup{div}}

\def\R{\mathbb{R}}

\def\F{\mathcal{F}}

\newcommand{\up}{\upshape}

\def\vv<#1>{\langle#1\rangle}

\newcommand{\tr}{\mbox{$\textup{Tr}$}}

\newcommand{\ev}{\mbox{$\text{\up{ev}}$}}

\newcommand{\dd}[2]{\mbox{$\frac{\partial #2}{\partial #1}$}}
\newcommand{\ddo}{\mbox{$\dd{t}{}|_{0}$}}
\providecommand{\del}{\partial}

\newcommand{\Om}{\Omega}

\newcommand{\lam}{\lambda}

\newcommand{\by}[2]{\mbox{$\frac{#1}{#2}$}}
\newcommand{\cinf}{\mbox{$C^{\infty}$}}

\newcommand{\X}{\mathfrak{X}}

\newcommand{\gu}{\mathfrak{g}}

\newcommand{\revise}[1]{#1}
\newcommand{\retwo}[1]{#1}

\title[]{%
Stochastic mean-field approach to fluid dynamics}

\author{Simon Hochgerner}
\email{simon.hochgerner@gmail.com} 
\date{October 28, 2017}

\begin{document}

\begin{abstract}
It is shown that the incompressible Navier-Stokes equation can be derived from an infinite dimensional mean-field stochastic differential equation.
\end{abstract}

\maketitle


\section*{Introduction} 
Let $M=T^n = \R^n/\mathbb{Z}^n$ be the $n$-dimensional torus.
The Navier-Stokes equation for incompressible viscous flow in $M$ with viscosity $\eta$ is
\begin{align*}
    \dd{t}{}u &= -\nabla_u u + \eta\Delta u - \nabla p\\
    \d\,u &= 0\\
    u(0,x) &= u_0(x)
\end{align*}
where $u = u(t,x)$ is the velocity of the fluid and $p=p(t,x)$ is the pressure.

The interpretation of viscous flow as a `sum' of  Euler flow and a stochastic flow has a long history (which we do not attempt to recount here). This philosophy has been applied by \cite{Y83,C,CC07} etc.\ to obtain a characterization of Navier-Stokes solutions as solutions to a stochastic variational problem. It has also led to numerical schemes where the Lie-Trotter product formula is applied to the Euler and the Stokes solution algorithms to obtain approximative solutions to the Navier-Stokes equations (\cite{Mar73,Chor78}). In \cite{Gli} solutions to the Euler equation are stochastically perturbed to construct solutions to the Navier-Stokes problem.  

Our approach \revise{enjoys the following features}:
\begin{enumerate}[\up (1)]
\item
We do not use the solution to any other (such as the Navier-Stokes or Euler) equation.
\item
Stochastic mean field equations are known to generally have non-linear generators, and this feature is used to reproduce the non-linerity in the Navier-Stokes equation.
For example, in \cite[Equation~(16.33)]{Gli} this is achieved through the use of an additional force term (that is subsequently annihilated by means of an additional equation).
\item
\revise{We use a stochastic} perturbation \revise{that} is driven by a cylindrical Wiener process in the infinite-dimensional space of divergence free vector fields (as in \cite{C,CC07}).
\item
The incompressibility condition $\d\,u=0$ is ensured through an extra (stochastic) pressure term. Our evolution equation is not formulated in a space of divergence free vector fields.
\item
In a sense our approach can be viewed as a Hamiltonian version of the Lagrangian formulation in \cite{Y83,C,CC07}. However this analogy has its limitations -- see Section~\ref{sec:Ham}.
\end{enumerate}

It is shown in Theorem~\ref{thm:main} that solutions to the Navier-Stokes equations can be constructed as expectations over solutions to an infinite dimensional mean-field stochastic differential equation. The idea is that the changes in the fluid's velocity field $u_t$ are due to a combination of a mean motion and a stochastic component. 

\revise{It turns out that points (1) and (2) are also satisfied by \cite{CI05}, the second item is stated explicitly. In fact,  the SDEs employed by \cite{CI05} are very similar to ours -- see Section~\ref{sec:lit} for a discussion.
}

\revise{Acknowledgements. The paper has benefited a lot from the two very insightful referee reports. Also I want to thank my colleagues at the FMA for discussing fluid mechanics over lunch time.}

\Section{Heuristics}\label{sec:heuristics}
\revise{The purpose of this section, which is not rigorous, is to serve as a motivation for equation~\eqref{e:sde2} in the next section. We remark that the heuristic equation \eqref{e:stovel} is certainly not new, but appears (in some form) in many works on stochastic fluid dynamics. 

Let $\gamma(t)$ denote the path of a fluid particle in $M$ and $\rho\xi$ the momentum along $\gamma$ where $\rho$ is the density of the fluid. We assume the density to be constant and set $\rho = 1$.
In line with \cite[Remark~5]{CHF17} we shall call $\xi$ the specific momentum of the fluid although, eventually, $E[\xi]$ is supposed to represent the (macroscopic) velocity of the fluid. In particular, we do not assume that $\dd{t}{}\gamma(t)$ and $\xi(t,\gamma(t))$ coincide.\footnote{The input of one of the referees concerning this point is gratefully acknowledged.}
 
Let the velocity $\dd{t}{}\gamma(t)$ consist of a mean component representing the average flow and a stochastic component that represents the random wandering of fluid particles. It is the latter that is responsible for the transfer of momentum between different layers of the fluid, and thus for turbulence. 
Hence we model the velocity as a sum of a smooth component $u$ and a Gaussian stochastic noise term $\nu\by{\delta W_t}{\delta t}$ such that
\begin{equation}\label{e:stovel}
 \dd{t}{}\gamma(t) = u(t,\gamma(t)) + \nu\frac{\delta W_t}{\delta t}
\end{equation}
with $E[\xi(t)] = u(t)$ and $\nu>0$. More precisely, we pick an orthonormal system of divergence free vector fields $X_{\alpha}$ and a sequence of independent real-valued Wiener processes $W^{\alpha}$ such that $W = \sum W^{\alpha}X_{\alpha}$, which is a cylindrical Wiener process in the space $\gu_0^s$ of divergence free vector fields of Sobolev class $H^s$ (technical details follow in the next section). Since $u$ is supposed to model the macroscopic fluid velocity we implicitly assume that $\dd{t}{}\gamma(t) = \dot{\gamma}(t)$ is divergence free.
}
According to Newton's second law we have
\[
 \by{d}{dt} \xi(t,\gamma(t)) = -\nabla p
\] 
where $-\nabla p$ is the pressure force. 
The usual balance of momentum argument, as in \cite{CM}, yields
\begin{align}\label{e:bom}
  \by{d}{dt} \xi(t,\gamma(t))
  &= \dd{t}{}\xi(t,\gamma(t)) + \vv<\dd{t}{}\gamma(t),\nabla>\xi(t,\gamma(t))\\
  &=
  \dd{t}{}\xi(t,\gamma(t))
  + \vv<u(t,\gamma(t))
  + \nu\frac{\delta W_t}{\delta t},\nabla>\xi(t,\gamma(t))
  \notag\\
  &= \dd{t}{}\xi(t,\gamma(t))
  + \nabla_u\xi(t,\gamma(t))
  + \nu\sum\nabla_{X_{\alpha}}\xi(t,\gamma(t))\, \frac{\delta W_t^{\alpha}}{\delta t}
  \notag\\
  &= -\nabla p
  \notag
\end{align}

Now think of $\xi_t = \xi(t,.)$ as a curve in  the space of vector fields on $M$.  Then we can rephrase the balance of momentum equation as the  Stratonovich SDE
\begin{align}\label{e:sde1}
\delta \xi_t
 &=
 -\nabla_{u_t}\xi_t \,\delta t
 - \nu\sum \nabla_{X_{\alpha}}\xi_t \,\delta W_t^{\alpha}
 - \nabla p \,\delta t \\
 \xi_0 &= u_0 \notag \\
E[\xi_t] &= u_t  \notag 
\end{align}
on the infinite dimensional space $\gu^s$ of vector-fields on $M$ of Sobolev class $H^s$ where $s$ will be specified below. 
This SDE depends on the mean $E[\xi_t] = u_t$ of the process, it is therefore a mean-field or McKean-Vlasov SDE.
Notice that \eqref{e:sde1} coincides with the Euler equation for incompressible flow ($\d\,u=0$ will be shown separately) if no stochasticity is present, i.e., $\nu=0$.  

The above equation is not formulated in the space of divergence free vector fields. So $\xi_t$ need not be divergence free. However, one could still hope that its mean satisfies $\d\,u_t=0$. To this end consider the pressure $p = p^a+p^m$ as a sum of an advective pressure $p^a$ due to the bulk motion of particles and another term $p^m$ which we think of as the molecular pressure created by the random wandering of particles. 
Let $PX = X-\nabla f$ be the Helmholtz projection operator such that $\d(PX)=0$. 
Then one should have $-\nabla p^a = (1-P)\nabla_u \xi$. The other term $\nabla p^m$ would then be responsible for annihilating $\d\,u$. It will be shown in the next section that such a function $p^m$ indeed exists whence the force in \eqref{e:bom} can be written as $-\nabla(p^a+p^m)$. Note that $p^a$ and $p^m$ both depend on $\xi_t$, i.e., they are random variables for given $x$ and $t$.
\revise{As stated above, the random fluctuation of individual particles is responsible for the transfer of momentum between individual sections of the fluid. Hence it is perhaps not surprising that $p^m$ will be proportional to $\d\,\xi$.}

If $\xi_t$ is a solution of \eqref{e:sde1} then the corresponding McKean-Vlasov equation for the law $\mu_t$ of $\xi_t$ reads
\[
\dd{t}{}\vv<\mu_t,\phi> = \vv<\mu_t,A^{\mu_t}\phi>
\]
where $A^{\mu_t}$ is the (non-linear) generator of the process $\xi_t$ and $\phi$ a sufficiently smooth function on $\gu^s$.
The bracket means evaluation of the probability measure $\mu$ on the function $\phi$: $\vv<\mu,\phi> = \int_{\gu^s}\phi(\xi)\,d\mu(\xi)$.
If we substitute $\phi = \ev_x: \gu^s\to\R^n$, $\xi\mapsto\xi(x)$ we have
\[
\vv<\mu_t,\ev_x> = E[\ev_x(\xi_t)] = u(t,x).
\]
It will be shown below that the McKean-Vlasov equation corresponding to \eqref{e:sde1} evaluated on $\phi = \ev_x$ yields
\[
\dd{t}{} u
= \dd{t}{} \vv<\mu_t,\ev_x>
= E[A^{\mu_t}\ev_x(\xi_t)]
= -P\nabla_uu + \eta\Delta u
\]
which, since it can be proved that also $\d\,u=0$, is the Navier-Stokes equation for incompressible flow in $M$ and where $\eta$ equals $\by{\nu^2}{2}$ up to a multiplicative constant.
The molecular pressure thus disappears from this equation which is due to the fact that $E[\nabla p^m]=0$ as will be shown in the next section.

\Section{Mean-field approach}
In the last section the Navier-Stokes equation was, heuristically, derived from the balance of momentum equation~\eqref{e:bom}. To make this precise we will start with the mean-field formulation~\eqref{e:sde1}. 

Let
$M=T^n=\R^n/\mathbb{Z}^n$. We fix $s>1+n/2$ and let $G$ denote the infinite dimensional $\cinf$-manifold of $H^s$-diffeomorphisms on $M$. 
Similarly, $G_0$ denotes the submanfiold of volume preserving $H^s$-diffeomorphisms. Both, $G$ and $G_0$, are  topological groups but not  Lie groups since left composition is only continuous but not smooth. Right composition is smooth.
The tangent space of $G$ (resp.\ $G_0$) at the identity $e$ shall be denoted by $\gu^s$ (resp.\ $\gu^s_0$). 
Let $\X^s(M)$ denote the vector fields on $M$ of class $H^s$ and $\X_0^s(M)$ denote the subspace of divergence free vector fields of class $H^s$.
We have $\gu_0^s = \X^s_{0}(M)$ and $\gu^s=\X^s(M)$.
See \cite{EM70,MEF}.

We equip $G_0$ with the $H^s$-metric and $\gu_0^s$ with the corresponding $H^s$-inner product. 
We choose an orthonormal system $X_{\alpha}$ of $\gu_0^s$ such that 
\[ 
 \nabla_{X_{\alpha}}X_{\alpha} = \vv<X_{\alpha},\nabla>X_{\alpha} = 0.
\]
Further we shall assume that, for $\xi\in\X(M)$, 
\begin{equation}\label{e:Delta}
 \sum \nabla_{X_{\alpha}}\nabla_{X_{\alpha}}\xi = c^s\Delta\xi
\end{equation}
where $c^s>0$ is a constant and $\Delta$ is the vector Laplacian.
Such an orthonormal system is explicitly constructed in \cite[Appendix]{CS09} (Lemma~20 with $g=e$). 

\revise{We define the (point) evaluation map by  $\ev_x: \gu^s\to\R^n$, $\xi\mapsto\xi(x)$. Since $\gu^s$ is a space of equivalence classes of measurable functions one could repeat the usual comments about this map being defined on a representative. Also, $\ev_x$ is defined on each $\gu^s$ in the same way, irrespectively of the smoothness class $H^s$.}

The following lemma plays the role of \cite[Theorem~2.2]{CC07} or \cite[Lemma~10]{CS09}. However, the statement is much simpler since we do not say anything about generators of diffusions.  

\begin{lemma}\label{lem:1}
Consider $\hat{X_{\alpha}}: \gu^s\to\gu^{s-1}$, $\xi\mapsto \nabla_{X_{\alpha}}\xi$. Then
\[
 \sum\hat{X_{\alpha}}\hat{X_{\alpha}}(\ev_x)(\xi) = c^s\Delta\xi.
\]
\end{lemma}

\begin{proof}
The directional derivative of $\ev_x$ along $\hat{X_{\alpha}}$ is not necessarily well-defined, since the latter is not a proper vector field on $\gu^s$ but takes values in $\gu^{s-1}$. However, since $\ev_x$ is linear the (Gateaux) derivative is
\[
 \hat{X_{\alpha}}(\ev_x)(\xi) 
 = d\ev_x(\hat{X_{\alpha}})(\xi)
 = \ddo\ev_x(\xi + t\hat{X_{\alpha}}(\xi))
 = \nabla_{X_{\alpha}}\xi(x)
\]
which clearly exists.
Similarly,
\begin{align*}
    \sum\hat{X_{\alpha}}\hat{X_{\alpha}}(\ev_x)(\xi)
    &= \sum d\big(\eta\mapsto\nabla_{X_{\alpha}}\eta(x)\big)(\xi).\hat{X_{\alpha}}(\xi)\\
    &= \sum\ddo\nabla_{X_{\alpha}}\big(\xi+t\hat{X_{\alpha}}(\xi)\big)(x)\\
    &= \sum\nabla_{X_{\alpha}}\nabla_{X_{\alpha}}\xi(x)\\
    &= c^s\Delta\xi(x)
\end{align*}
where we make use of \eqref{e:Delta}.
\end{proof}

Let $(\Om,\F,(\F_t)_{t\in[0,T]},P)$ be a filtered probability space satisfying the usual assumptions. In the following, Brownian motion shall be understood to be adapted to this filtration. 

The following proposition (where the second statement is only cited for general information) is well-known -- see e.g.\ \cite{CC07,C,DZ}. 

\begin{proposition}
  The following are true.
  \begin{enumerate}[\up (1)]
  \item
    Let $W_t = \sum X_{\alpha} W_t^{\alpha}$, where $W_t^{\alpha}$ are independent copies of Brownian motion in $\R$. Then $W$ defines (a version of) Brownian motion (or cylindrical Wiener process) in $\gu_0^s$.
  \item
    If $\gamma_t$ satisfies $\delta \gamma_t = T_e R_{\gamma_t}.\delta W_t$, then $\gamma_t$ is Brownian motion in $G_0^s$ with respect to the induced right-invariant $H^s$-metric.
  \end{enumerate}
\end{proposition}

Let us introduce the Helmholtz projection operator (see \cite[Corollary 1.4.4]{MEF}):
Consider a $H^s$ vector field $X$ on $M$. Then there is a unique divergence free vector field $Y$ of class $H^s$ and a function $f$ on $M$ such that $X=Y+\nabla f$. Setting $Y=PX$ thus defines a bounded linear operator $P: \X^s(M)\to\X_0^s(M)$ (for arbitrary $s\ge0$).
      
Motivated by \eqref{e:sde1} we consider, for $\eta>0$ and $\nu=\sqrt{\by{2\eta}{c^s}}$,  the mean-field Stratonovich equation in $\gu^s$
\begin{align}\label{e:sde2}
 \delta \xi_t 
 &= - \Big(P\nabla_{u_t}\xi_t + \eta\nabla\d\,\xi_t \Big)\,\delta t
    - \nu \sum \nabla_{X_{\alpha}}\xi_t \,\delta W_t^{\alpha}\\
 \xi_0 &= u_0  \notag \\
 u_t &= E[\xi_t] \notag
\end{align}
where $u_0\in\gu_0^s$ is the initial condition. Note that this is an equation in all of $\gu^s$. Its solutions are not necessarily divergence free.  
The term $-\eta\nabla\d\,\xi_t$ cannot be guessed from \eqref{e:sde1}. Its role will become clear below. Note that $\nabla\d\,\xi = (1-P)\Delta\xi$.

 
\begin{definition}[Strong solution]
A $\gu^s$-valued stochastic process $\xi_t$ with $t\in[0,T]$ is a strong solution to the mean-field or McKean-Vlasov Stratonovich SDE 
\[
 \delta \xi_t = f(\xi_t,\mu_t)\,\delta t + g(\xi_t)\,\delta W_t
\] 
if
\begin{enumerate}[\up (1)]
\item
$\mu_t$ is the law of $\xi_t$, i.e., $\mu_t = (\xi_t)_*P$,
\item
$\xi_t$ is adapted to $(\F_t)$,
\item
$t\mapsto\xi_t$ is continuous $P$-a.s.,
\item
$\xi_t 
 = \xi_0 + \int_0^t f(\xi_s,\mu_s) \,ds 
   + \int_0^t g(\xi_s)\,\delta W_s $ 
for all $t\in[0,T]$ $P$-a.s.
\end{enumerate}
\end{definition}

Notice that, once the prescription $t\mapsto\mu_t$ of the law of $\xi_t$ is found, the concept of a mean field Stratonovich SDE is not different from that of a time-dependent Stratonovich SDE. The Stratonovich integral above is hence, as usually, to be understood as 
\[
 \int_0^t g(\xi_s)\,\delta W_s
 = \int_0^t g(\xi_s)\, d W_s
 + \by{1}{2}[g(\xi),W]_t.
\]
See e.g.\ \cite[Page~82]{Pro}.


\begin{theorem}\label{thm:main}
If $\xi_t\in\gu^s$ is a strong solution to \eqref{e:sde2} on $[0,T]$ such that $\xi_0=u_0\in\gu_0^s$ then $u(t,x) = E[\ev_x(\xi_t)]$ satisfies the Navier-Stokes equation with $u(0,x)=u_0(x)$ for incompressible flow in $M$ on $[0,T]$:
\begin{align*}
    \dd{t}{} u &= -P\nabla_u u + \eta\Delta u \\
    \d\,u &= 0
\end{align*}
\end{theorem} 

Note that there are no boundary conditions on $M$ and $\d\,u=0$, whence $P\Delta u = \Delta Pu = \Delta u$.

\begin{proof}
Let us write the Ito version of \eqref{e:sde2} as $d\xi_t = b(u_t,\xi_t)\,dt - \nu\sum\nabla_{X_{\alpha}}\xi_t\,dW_t^{\alpha}$.
We have then for the quadratic variation
\begin{align*}
\sum\Big[-\nu\nabla_{X_{\alpha}}\xi, W^{\alpha}\Big]_t 
&= -\nu\sum\Big[
    \int_0^{\cdot}\nabla_{X_{\alpha}} b(u_s,\xi_s)\, ds 
    -
    \nu\int_0^{\cdot}\nabla_{X_{\alpha}}\nabla_{X_{\beta}}\xi_s\,dW_s^{\beta},
    W^{\alpha}
 \Big]_t \\
& = 
\nu^2\sum\delta_{\alpha\beta}\int_0^t          \nabla_{X_{\alpha}}\nabla_{X_{\beta}}\xi_s\;ds \\
&= 
c^s\nu^2\int_0^t \Delta\xi_s\, ds
\end{align*}
where $\delta_{\alpha\beta}$ is the Kronecker delta.
(See \cite{Pro} or \cite[Section~3.4]{DZ}.) 
Therefore, the Ito version of \eqref{e:sde2} is
\begin{equation}\label{e:sed2-ito}
d \xi_t 
 = \Big(-P\nabla_{u_t}\xi_t - \eta\nabla\d\,\xi_t 
         + c^s\by{\nu^2}{2}\Delta\xi_t \Big)\,dt
    - \nu \sum \nabla_{X_{\alpha}}\xi_t \, dW_t^{\alpha}
\end{equation}
To see that 
\[
 u_t = u_0 
     + E\Big[\int_0^t\Big(-P\nabla_{u_s}\xi_s - \eta\nabla\d\,\xi_s 
         + \eta\Delta\xi_s \Big)\,ds \Big]
\] 
is divergence free we note that $\d\,u_0 = 0$ by assumption and the integrand is sufficiently smooth such that 
\begin{align*}
    \d\,u_t
    &= E\Big[\int_0^t \d \Big(-P\nabla_{u_s}\xi_s - \eta\nabla\d\,\xi_s 
         + \eta\Delta\xi_s \Big)\,ds \Big] \\
    &=  E\Big[\int_0^t \Big( 0 - \eta\Delta\d\,\xi_s 
         + \eta\d \Delta\xi_s \Big)\,ds \Big] \\
    &= 0.
\end{align*}

As in Lemma~\ref{lem:1}, for $t\in[0,T]$, consider the $\gu^{s-1}$-valued vector fields  $\hat{X_{\alpha}}$.
Additionally we define the time dependent $\gu^{s-2}$-valued vector field
$\hat{X_t}$ on $\gu^s$ by
$\hat{X_t}(\xi) = - P\nabla_{u_t}\xi - \eta\nabla\d\,\xi$.

Because of the Ito formula applied to $\ev_x: \gu^s\to\R^n$ the process $\ev_x(\xi_t)$ satisfies
\begin{align*} 
E\Big[\ev_x(\xi_t)\Big]
    & = u_0(x) +
        E\Big[\int_0^t\Big(\hat{X_s} + \by{\nu^2}{2}\sum\hat{X_{\alpha}}\hat{X_{\alpha}}\Big)(\ev_x)(\xi_s)\,ds\Big] \\
    &= u_0(x) +
        E\Big[\int_0^t\Big(\hat{X_s}(\ev_x)(\xi_s) + c^s\by{\nu^2}{2}\Delta\xi_s(x)\Big)\,ds\Big]
\end{align*}
for all $x\in M$ and where we have used Lemma~\ref{lem:1}.
By linearity $\hat{X_t}(\ev_x)(\xi) = d\ev_x(\hat{X_t})(\xi) = \ev_x(\hat{X_t}(\xi)) = -P\nabla_{u_t}\xi(x) - \eta\nabla\d\,\xi(x)$.
Hence, for $u(t,x) = u_t(x) = E[\ev_x(\xi_t)]$, 
\begin{align*}
 \dd{t}{} u(t,x) 
 &=
  -P\nabla_{u_t}E[\xi_t](x)
  - \eta\nabla\d\,E[\xi_t](x)
  + \eta\Delta E[\xi_t](x)
 = -P\nabla_{u_t}u_t(x) + \eta\Delta u_t(x)
\end{align*}
since $\d\,u_t = 0$.
\end{proof}

The above proof shows that, while $\xi_t$ is not divergence free, the effect of the  $\eta\nabla\d\,\xi_t$ term is to make $\d\,\xi_t$ into a martingale. 
Returning to the discussion of the previous section we have thus found $p^m = \eta\d\,\xi_t$, up to a constant.

\Section{Conclusions and outlook}
\revise{\subsection{Comparison with existing literature}\label{sec:lit}
While the derivation of the Navier-Stokes equation via equation~\eqref{e:sde2} seems to be new, the results and also the interpretation as a stochastic mean-field system are not new.\footnote{It is gratefully acknowledged that one of the referees pointed out \cite{CI05} (which was unknown to the author).}
Indeed, Constantin and Iyer~\cite{CI05} explicitly mention the non-linearity in the sense of McKean in their work. Furthermore, they 
\retwo{define the stochastic vector $w = A'\otimes(u_0\circ A)$
where $A=X^{-1}$, $X=X(t,x)$ is the Lagrangian flow map, i.e.\ a volume preserving diffeomorphism for each $t$, and the divergence free vector field $u_0$ is the initial condition. 
We make use of the notation $A'\otimes v = \sum(\del_i A^j)v^j e_i = (\nabla^t A)v$ where $(e_i)$ is the standard basis in $\R^3$.
Let $B$ denote three-dimensional Brownian motion. 
In \cite[Theorem~2.2]{CI05} it is assumed that the pair $(X,u)$ satisfies
\begin{align*}
    dX &= u\,dt + \sqrt{2\eta}\,dB\\
    u &= E[Pw]\\
    X(0,x) &= x
\end{align*}
and conluded that $u$ is a solution of the incompressible Navier-Stokes equation with initial data $u_0$ and viscosity coefficient $\eta$.
In order to do so they}
derive the system
\begin{equation}
\tag{\cite[Equation~4.5]{CI05}}
dw_t = \Big(-\nabla_{u_t} w_t + \eta \Delta w_t - u_t'\otimes w_t\Big)\,dt + \sqrt{2\eta}\nabla w_t\,d\retwo{B}_t
\end{equation}
where $(w,u)$ play the role of our $(\xi,u)$. 
Apart from the fact that this equation is finite dimensional it looks very much like \eqref{e:sed2-ito} above. Basically the $X_{\alpha}$ are replaced by $e_i$. Notice also that $w$, just like $\xi$ above, is not divergence free.  
Now, \cite{CI05} construct the mean Eulerian velocity as $u = PE[w]$. In order to be comparable to our equation~\eqref{e:sed2-ito} we should consider $u = E[Pw]$, which is completely equivalent. Thus
\begin{align*}
    u_t 
    &= u_0 + 
    E\Big[
     \int_0^t P\Big( -\nabla_{u_s} w_s + \eta\Delta w_s - u_s'\otimes w_s\Big)\,dt
    \Big].
\end{align*}
Let us look at the individual terms. The first, $-P\nabla_{u_t} w_t$, is obvious. The second is $P\eta\Delta w_t = \eta\Delta w_t - \eta\nabla\d\,w$, which we recognize as-well. 
For the third we note that 
\[
 \retwo{\Delta}\vv<u,w> = 
 \vv<\Delta u,w> + 2\tr(u'\otimes w') + \vv<u,\Delta w>
\]
and 
\[
 \d(u'\otimes w) = \vv<\Delta u, w> + \tr (u'\otimes w').
\]
Hence $EP(u'\otimes w) = PE(u'\otimes w) = u'\otimes u - \by{1}{2}\nabla\vv<u,u>$, and the last term appears just below \cite[Equation~4.5]{CI05} as it should. However, more is true: if we add the stochastic gradient $-\by{1}{2}\nabla\vv<u,w>$ to the above equation we obtain 
\[
 u'\otimes w - \by{1}{2}\nabla\vv<u,w>
 = \by{1}{2}\Big(w'\otimes u - u'\otimes w\Big)
\]
which certainly vanishes upon taking expected values.  
Thus the drift terms of \cite[Equation~4.5]{CI05} and \eqref{e:sed2-ito} agree up to a gradient and a remainder that vanishes in the mean. Of course, otherwise the respective expectations could not satisfy the same equation. 
Thus our approach differs/agrees with that of Constantin and Iyer \cite{CI05} in the following aspects:
\begin{enumerate}[\up (1)]
\item
We use a (Stratonovich) infinite dimensional set-up via a cylindrical Wiener process as opposed to finite dimensional Brownian motion.
\item
The vanishing divergence of the expectation $u$ is a direct consequence of the equation for $\xi$. (However, as shown above, the approach \cite{CI05} approach can easily be modified to yield the same feature.)
Further, the drifts of the respective SDEs are very similar, albeit not identical.
\item
Both papers use a stochastic mean-field formulation.
\item
The respective physical interpretations are also very similar. One proceeds via a stochastic Lagrangian particle formulation and an averaged Weber formula, the other via a stochastic Newton equation (along stochastic particle paths). 
\end{enumerate}
}

\subsection{Hamiltonian approach}\label{sec:Ham}
Let $\gu=\X(M)$ denote the infinite dimensional Frechet space of $\cinf{}$ vector fields on $M$. This space has a natural Lie-Poisson structure and the Hamiltonian vector field associated to the function $\xi\mapsto\vv<u,\xi>$ is $-\nabla_u\xi$.
In a sense our approach can thus be viewed as a Hamiltonian version of the Lagrangian formulation in \cite{Y83,C,CC07}. 
However, this analogy has its limitations since, instead of $\gu$, we are dealing with $\gu^s$ and \eqref{e:sde2} is not an equation with respect to any (natural) Lie-Poisson structure on $\gu^s$: 
\begin{itemize}
    \item
    $\nabla_{X_{\alpha}}\xi\in\gu^{s-1} \neq \gu^s$;
    \item
    $P\nabla_{u}\xi\in\gu_0^{s-1}\neq \gu^s$.
\end{itemize}
\revise{Stochastic Hamiltonian versions of the Navier-Stokes and other equations of fluid dynamics equation are also contained with more details in \cite{AB10,E10,ED16,Rez}.}
\retwo{In fact, \cite{AB10} present a stochastic Lagrangian formulation of the Navier-Stokes equation that is very similar to that of \cite{CI05}. In \cite[Section~6]{AB10}, and along the lines of \cite{C00}, it is further shown how the Newton equation ties to the Weber formula used by \cite{CI05}. This point of view thus also provides a link between our Section~\ref{sec:heuristics} and the stochastic Weber formula approach. 
}

\subsection{Independent particle approximation}
We refer to \cite{AD95,Gov} and the references contained therein for further information on stochastic McKean-Vlasov or mean field equations in a Hilbert space setting. 

Let us write equation~\eqref{e:sed2-ito}, which is the Ito version of our basic equation~\eqref{e:sde2}, 
as $d\xi_t = b(\xi_t,\mu_t)\,dt + \sigma(\xi_t)\,dW_t$
where $\mu_t$ is the law of $\xi_t$.
This equation has an approximation by a system of independent particles. For $k=1,\dots,N$
consider the coupled stochastic evolution equations
\begin{align}\label{e:ips1}
d\xi_t^{k,N}
&= b(\xi_t^{k,N}, \mu_t^N)\,dt + \sigma(\xi_t^{k,N})\,dW_t\\
\mu_t^N &= \by{1}{N}\sum_{j=1}^N \delta_{\xi_t^{j,N}} \notag\\
\xi_{0}^{k,N} &= u_0 \notag
\end{align}
where $\delta_{\xi}$ is the Dirac distribution in $\gu^s$ centered at $\xi$. The measure $\mu_t^N$ is called the \emph{empirical measure}. 

The results of \cite{AD95,Gov} cannot directly be applied to the system \eqref{e:ips1} since $b$ and $\sigma$ do \emph{not} take values in $\gu^s$. (See also the comments in Section~\ref{sec:locsol}.)
Let $M_{\lam^2}(\gu^s)$ denote an appropriate subset of the set of probability measures on $\gu^s$ as defined in \cite{AD95}. Then, if $b: \gu^s\times M_{\lam^2}(\gu^s)\to\gu^s$ and $\sigma:\gu^s\to\gu^s$ satisfy certain Lipschitz and  linear growth conditions, it follows from loc.\ cit.\ that $\mu_t^N$ converges in probability, as $N\to\infty$, to the (deterministic) probability measure $\mu_t$ which is the law of the solution $\xi_t$ to \eqref{e:sed2-ito}.   
Moreover, $\xi_t$ is the limit (in the appropriate sense) of $\xi_t^{1,N}$ as $N\to\infty$.

One way to make all of this precise could be to employ the mollification technique of \cite{MB03} to $b$ and $\sigma$. This will be a topic for future research.

\subsection{Existence of local solutions}\label{sec:locsol}
The local existence results of \cite{AD95,Gov} cannot be applied to \eqref{e:ips1} since this system contains unbounded operators. The situation is very similar to case of the Navier-Stokes or Euler equations which also do not allow for a direct application of the Picard theorem on Banach spaces. 

Local existence and uniqueness of solutions to the incompressible Navier-Stokes and Euler equations are shown in \cite[Section~3]{MB03} by means of mollification and energy bounds. It would be interesting to carry out the same programme for the mean-field equation \eqref{e:sde2} or the system \eqref{e:ips1}. One of the challenges would be to find the appropriate stochastic energy bounds. 
\revise{As a matter of fact this has actually been accomplished by \cite{I06a,I06b} for the Weber formula approach of \cite{CI05}.}

\subsection{Numerics}
It is an obvious idea to apply one of the existing numerical schemes for stochastic mean-field equations (see e.g.\ \cite{delmoral}) to the system \eqref{e:ips1} in order to construct numerical solutions to \eqref{e:sde2}. By construction this approach yields a Monte-Carlo method for approximate solutions to the Navier-Stokes equations:
\begin{enumerate}[\up (1)]
\item
Take a large $N$.
\item
Discretize time and space to approximately solve \eqref{e:ips1}.
\item
Take the sample average over the approximate solution for $\xi_t^{1,N}$.
\end{enumerate}
It is of course not clear at all whether this approach would have any advantage compared to existing algorithms.  
In particular it would be interesting to numerically study the effect of the $\eta\nabla\d\,\xi_t$ term.

\subsection{Physical interpretation}
The derivation of the incompressible Navier-Stokes equations via equation~\eqref{e:sde2} and Theorem~\ref{thm:main} is, to the best of our knowledge, new. It is interesting to consider the physical meaning of the individual terms in equation~\eqref{e:sde2}:
\begin{itemize}
    \item 
    The mean-field $u_t$ appears in the advection term $P\nabla_{u_t}\xi_t$. This is reasonable as advection does not take place for a single particle but is always a property of bulk motion. Thus the interpretation of the mean-field is that of the independent particle system \eqref{e:ips1} such that individual fluid particles are transported by the stream of other particles.
    \item
    On the other hand, \revise{momentum is transferred between adjacent layers of fluid due to the random wandering of individual particles. This is modelled by the $\nu\nabla_{X_{\alpha}}\xi_t\,\delta W_t^{\alpha}$ term.
    This stochastic transfer of momentum has the consequence that $\xi_t$ is not divergence free. However, the exchange of momentum results in an additional pressure term $-\eta\nabla\d\xi_t$ whose effect is to cancel the divergence in the mean: $E[\d\,\xi_t] = \d\,u_t = 0$.} 
    \item
    The stochasticity in \eqref{e:sde2} is not just due to Brownian motion in the space dimensions but given by a Wiener process in the space of divergence free vector fields. This means that we have a stochastic perturbation by all possible velocity fields. Again, this makes sense, since the idea is that the perturbation is due to particles which have wandered from different velocity layers. 
\end{itemize}

\end{document}